\documentclass[fleqn,11pt,twoside]{article}

\usepackage{amsthm,amsthm,amssymb, color, xcolor,epsfig, graphics}

\usepackage{amsmath, graphicx, latexsym, lscape}

\makeatletter
\newcommand{\copyrightnote}[2]{{\renewcommand{\thefootnote}{}
 \footnotetext{\small\it
\begin{flushleft}
 \copyright \ #1   #2  
\end{flushleft}}}}

\newcommand{\Name}[1]{\begin{flushleft}
                       \LARGE \bf #1
                       \end{flushleft}\vspace{-3mm}}

\newcommand{\Author}[1]{\begin{flushleft}
                       \it #1 \end{flushleft}}

\newcommand{\Address}[1]{\begin{flushleft}
                       \it #1 \end{flushleft}}

\newcommand{\Date}[1]{\begin{flushleft}
                      \small  \it #1 \end{flushleft}}

\newcommand{\evenhead}{Author \ name}
\newcommand{\oddhead}{Article \ name}
\newcommand\cH{{\mathcal H}}
\newcommand\cS{{\mathcal S}}
\newcommand\cA{{\mathcal A}}
\newcommand\cT{{\mathcal T}}
\newcommand{\field}[1]{\mathbb{#1}}
\newcommand{\C}{\field{C}}
\newcommand{\Z}{\field{Z}}
\newcommand{\N}{\field{N}}

\DeclareMathOperator{\Ima}{Im}

\newtheorem{theorem}{Theorem}[section]

\newtheorem{corollary}{Corollary}[section]

\renewcommand{\@evenhead}{
\hspace*{-3pt}\raisebox{-15pt}[\headheight][0pt]{\vbox{\hbox to \textwidth
{\thepage \hfil \evenhead}\vskip4pt \hrule}}}
\renewcommand{\@oddhead}{
\hspace*{-3pt}\raisebox{-15pt}[\headheight][0pt]{\vbox{\hbox to \textwidth
{\oddhead \hfil \thepage}\vskip4pt\hrule}}}
\renewcommand{\@evenfoot}{}
\renewcommand{\@oddfoot}{}

\long\def\@makecaption#1#2{%
  \vskip\abovecaptionskip
  \sbox\@tempboxa{\small \textbf{#1.}\ \ #2}%
  \ifdim \wd\@tempboxa >\hsize
    {\small \textbf{#1.}\ \ #2}\par
  \else
    \global \@minipagefalse
    \hb@xt@\hsize{\hfil\box\@tempboxa\hfil}%
  \fi
  \vskip\belowcaptionskip}

\newcommand{\JNMPnumberwithin}[3][\arabic]{%
  \@ifundefined{c@#2}{\@nocounterr{#2}}{%
    \@ifundefined{c@#3}{\@nocnterr{#3}}{%
      \@addtoreset{#2}{#3}%
      \@xp\xdef\csname the#2\endcsname{%
        \@xp\@nx\csname the#3\endcsname .\@nx#1{#2}}}}%
}

\renewenvironment{proof}[1][\proofname]{\par
  \normalfont
  \topsep6\p@\@plus6\p@ \trivlist
  \item[\hskip\labelsep\textbf{%
    #1\@addpunct{.}}]\ignorespaces
}{%
  \qed\endtrivlist
}

\newcommand{\resetfootnoterule} {
  \renewcommand\footnoterule{%
  \kern-3\p@
  \hrule\@width.4\columnwidth
  \kern2.6\p@}
}

\renewcommand{\footnoterule}{}

\makeatother

\theoremstyle{definition}

\begin{document}

\renewcommand{\evenhead}{ {\LARGE\textcolor{blue!10!black!40!green}{{\sf \ \ \ ]ocnmp[}}}\strut\hfill E Peroni and J P Wang}
\renewcommand{\oddhead}{ {\LARGE\textcolor{blue!10!black!40!green}{{\sf ]ocnmp[}}}\ \ \ \ \   Hamiltonian and recursion operators discrete
Kaup-Kupershmidt}

%%%% Matter for the first page 
\thispagestyle{empty}
\newcommand{\FistPageHead}[3]{
\begin{flushleft}
\raisebox{8mm}[0pt][0pt]
{\footnotesize \sf
\parbox{150mm}{{Open Communications in Nonlinear Mathematical Physics}\ \ \ \ {\LARGE\textcolor{blue!10!black!40!green}{]ocnmp[}}
\quad Special Issue 1, 2024\ \  pp
#2\hfill {\sc #3}}}\vspace{-13mm}
\end{flushleft}}

\FistPageHead{1}{\pageref{firstpage}--\pageref{lastpage}}{ \ \ }

\strut\hfill

\strut\hfill

\copyrightnote{The author(s). Distributed under a Creative Commons Attribution 4.0 International License}

\begin{center}
%{\Large  {\sf This article is part of a Special Issue in Memory of Professor Decio Levi}}
{  {\bf This article is part of an OCNMP Special Issue\\ 
\smallskip
in Memory of Professor Decio Levi}}
\end{center}

\smallskip

\Name{Hamiltonian and recursion operators for a discrete analogue of the Kaup-Kupershmidt equation}

\Author{Edoardo Peroni and Jing Ping Wang}

\Address{School of Mathematics, Statistics \& Actuarial Science, University of Kent, UK
}

\Date{Received July 5, 2023; Accepted November 29, 2023}

\setcounter{equation}{0}

\begin{abstract}

\noindent 
In this paper we study the algebraic properties of a discrete analogue of the Kaup-Kupershmidt equation. This equation can be seen as a
deformation of
		the modified Narita-Itoh-Bogoyavlensky equation and has the
Kaup-Kupershmidt equation in its continuous limit.
Using its Lax representation we explicitly construct a recursion operator for this equation
 and prove that it is a Nijenhuis operator. Moreover,
we present the bi-Hamiltonian structures for this equation.

\end{abstract}

\label{firstpage}

%%%% The Article text starts here

\section{Introduction}
In a recent paper \cite{15}, Mikhailov, Novikov and Wang formulated necessary
integrability conditions for evolutionary differential-difference equations (D$\Delta$Es) using the perturbative symmetry approach and
obtained a family of new integrable D$\Delta$Es:
\begin{eqnarray}\label{squareq}
u_t=(1+u^2)(u_p\prod_{k=1}^{p-1}\sqrt{1+u_k^2}-u_{-p}\prod_{k=1}^{p-1}\sqrt{1+u_
{-k}^2}), \quad  p\in \mathbb{N},
\label{initialeq}
\end{eqnarray}
where $u=u(n,t)$ is a complex function of one discrete variable $n\in\Z$ and one
continuous independent variable $t$. Here
we use the standard notations
\begin{eqnarray*}
 u_t=\partial_t(u), \quad u_k=\cS^k u(n,t)=u(n+k,t), \quad k\in \Z
\end{eqnarray*}
and $\cS$ is the shift operator.
Equation (\ref{squareq}) can be viewed as an inhomogeneous deformation of the
Narita-Itoh-Bogoyavlensky (NIB) equation \cite{30}. Indeed, after the re-scaling
\begin{equation}\label{scaling}
u_k\to \varepsilon^{-1} u_k,\quad\ t\to \varepsilon^{p+1}t
\end{equation}
and in the limit $\varepsilon\to
0$ it turns into
\begin{equation}
u_t=u^2(\prod_{k=1}^{p} u_k-\prod_{k=1}^{p} u_{-k})  \label{mNIB}.
\end{equation}
The authors of \cite{15} have proved the integrability of (\ref{squareq}) for any
$p\in \mathbb{N}$ by constructing its Lax representation and provided the $r$
matrix approach to produce the hierarchy of commuting symmetries. They showed
that
equation (\ref{squareq})
possesses a Lax representation  $L_t=[A, L]$ with a rational (pseudo-difference)
operator $L=Q^{-1}P$, where
\begin{eqnarray}\label{lax}
 Q=u - u_1 w \cS^{-1}, \quad P=(u w_1 \cS-u_1) \cS^{p-1},\qquad  w=\sqrt{1+u^2},
\end{eqnarray}
and a skew-symmetric difference operator $A=L_{+}-(L_{+})^{\dagger}$. Here $L_+$
denotes the polynomial part of the Laurent series
for the rational operator $L$ and $\dagger$ denotes the adjoint of an operator. The
hierarchy of commuting symmetries for \eqref{squareq} can be generated by
\begin{eqnarray}\label{symp}
 L_{t_l}=[\pi_{+}(L^l), L], \qquad
\pi_{+}(L^l)=\left(L^l\right)_{+}-\left(\left(L^l\right)_{+}\right)^{\dagger},\quad l\in \N,
\end{eqnarray}
where $t_1=t$ corresponds to the equation itself.

For every member $p$ of the family \eqref{initialeq}, if we use the substitution
\begin{equation}\label{conlp}
 \begin{split}
  & u(n,t) = \frac{2 \sqrt{p}}{p+1}{\rm{i}} - \frac{(p-1)^2 }{12  (p+1)} \sqrt{p}\ {\rm{i}}\epsilon^2 U\left(x + A(p) \epsilon t, \tau + B(p) \epsilon^5 t\right),\qquad  x=\epsilon n, \\
  & A(p) = -2 p \left(\frac{p-1}{p+1}\right)^{p}, \quad  B(p) = \frac{p^3 (p^2-1)}{180}  \left(\frac{p-1}{p+1}\right)^{p},
 \end{split}
\end{equation}
where ${\rm i}^2=-1$ is the imaginary unit, we always obtain the Kaup-Kupershmidt
equation
\begin{equation*}
		U_\tau = 5 U^2 U_x + \frac{25}{2}  U_x U_{xx} + 5 U U_{xxx} +
U_{xxxxx}\label{KK}
	\end{equation*}
	for a function $U(\tau, x)$ as $\epsilon\rightarrow 0$. This is analogous to what happens for the NIB family, whose members have a continuous limit in the Korteweg-de Vries equation.

This paper is devoted to the study of recursion operators and Hamiltonian
structures of the equation \eqref{squareq} for the case when $p=2$, that is,
	\begin{equation}
		u_t =(1+u^2) (u_{2}
\sqrt{1+u_{1}^2}-u_{-2}\sqrt{1+u_{-1}^2}),
		\label{1.1}
	\end{equation}
which was first obtained in the symmetry classification of a class
of five-point differential-difference equations \cite{49}. Its integrability
was established
in \cite{50}, where Garifullin and Yamilov constructed a Lax representation for
(\ref{1.1}), which is in a different form from (\ref{lax}) when $p=2$ (see Remark
3 in \cite{15} for the comparison).
As shown also in \cite{50}, the continuous limit of \eqref{1.1} is the Kaup-Kupershmidt
equation using the substitution
\begin{equation*}
		u(n, t) = \frac{2\sqrt{2}}{3} {\rm{i}} - \frac{\rm{i}}{18
\sqrt{2}} \epsilon^2 U\left(x - \frac{4}{9} \epsilon t, \tau
+\frac{2}{135}\epsilon^5 t \right), \quad x=\epsilon n,
	\end{equation*}
which is the case of (\ref{conlp}) when $p=2$.

From the formulas (\ref{symp}), we are able to explicitly compute its second member of the hierarchy for $p=2$, namely,
	\begin{equation}
		\begin{split}
			&u_{t_2} =w^2 (u_4 w_1 w_2^2 w_3+u_1 u_2 u_3 w_1 w_2+u u_2^2 w_1^2+u_{-1} u_1 u_2 w w_1 \\& \qquad
			-u_{-2} u_{-1} u_1 w_{-1} w-u_{-2}^2 u w_{-1}^2  -u_{-3} u_{-2} u_{-1} w_{-2} w_{-1}-u_{-4} w_{-3} w_{-2}^2 w_{-1}),
		\end{split} \label{5.9}
	\end{equation}
which can also be obtained via $u_{t_2}=\mathfrak{R} \left(u_t \right)$.
Here $\mathfrak{R}$ is defined by
	\begin{eqnarray}
			&& \mathfrak{R} =
\mathfrak{R}^{(1)}+\mathfrak{R}^{(3)},\label{RecursionInIntro}\\
			&& \mathfrak{R}^{(1)} =w^2 \left(\left( {u_1}- \mathcal{S}^{-1}
u\right) \left(\mathcal{S}w - w \mathcal{S}^{-1}\right)^{-1} \left(  u
\mathcal{S} + u_1\right)+\mathcal{S}w \mathcal{S}
+\mathcal{S}^{-1} w \mathcal{S}^{-1}\right)\frac{1}{w^2};\nonumber \\
			&& \mathfrak{R}^{(3)} =w^2\left(\left(
{u_1}-\mathcal{S}^{-1} u \right) \left(\mathcal{S}w - w
\mathcal{S}^{-1}\right)^{-1} \left(u \mathcal{S}-u_1\right) + \mathcal{S}
w\mathcal{S} -  \mathcal{S}^{-1}w \mathcal{S}^{-1}\right)\nonumber\\
&&\qquad \quad u\left(\mathcal{S}^2+\mathcal{S }+1\right) (\mathcal{S}^2-1)^{-1}\frac{u}{w^2} \nonumber
	\end{eqnarray}
In this paper we are going to prove that the operator $\mathfrak{R}$ is a Nijenhuis recursion operator for \eqref{1.1}. Note that under the same rescaling (\ref{scaling}) and the
limit in $\epsilon$, the operator \eqref{RecursionInIntro} goes to
\begin{equation*}
 u^2 \mathcal{S}^{-1}\left(\mathcal{S}^3 -1\right) u \left( \mathcal{S} u- u
\mathcal{S}^{-1}\right)^{-1} u\left(1 - \mathcal{S}^{-3}\right)
 \mathcal{S} u\left(\mathcal{S}^2+\mathcal{S}+1\right)\left(u \mathcal{S}^2-u
\right)^{-1},
	\end{equation*}
which is the recursion operator of \eqref{mNIB} for $p=2$ in \cite{7}
reflecting the fact that equation \eqref{1.1} is its inhomogeneous
deformation.
The derivation of $\mathfrak{R}$ is based on the structures of Lax
representations for the symmetry hierarchy proposed in
\cite{48}
for partial differential equations. This approach was soon adapted for
D$\Delta$E in \cite{58} and further developed in \cite{7,1} taking into account
the reduction groups of Lax representations \cite{31}.
The detailed derivation will be described in Section \ref{Sec2}, where
we also prove that the operator $\mathfrak{R}$ is Nijenhuis since it can
be written as a rational operator of a preHamiltonian pair \cite{4, 41}.
In the same section, we represent the recursion operator in a new form
involving the action of the reflection operator $\mathcal{T}$, which is also an automorphism defined by
$\mathcal{T} u_n = u_{-n}$, that is,
	\begin{equation}\label{RT}
		\begin{split}
			\mathfrak{R}' & = w^2 \left(1-\mathcal{T}\right) \Bigg[
u_1 \left( \mathcal{S} w - w \mathcal{S}^{-1} \right)^{-1} \left( \left( u
\mathcal{S}+u_1 \right)\frac{1}{u}\left( \mathcal{S}^2-1\right) + u
u_1\left(\mathcal{S}^3-1\right) \right)+ \\ & + w_1 \mathcal{S}^2 \left(\frac{1}{u}\left(\mathcal{S}^2-1\right)+u
\left(\mathcal{S}^2+\mathcal{S}+1\right)\right) \Bigg]
(\mathcal{S}^2-1)^{-1}\frac{u}{w^2} .
		\end{split}
	\end{equation}
The presence of the operator $\left(1- \mathcal{T}\right)$
in the expression above reveals that all the symmetries in the hierarchy are odd
with respect to $\mathcal{T}$. Mathematically, operators $\mathfrak{R}$ and $\mathfrak{R}'$
are not equal while they give the same values on the image of $\left(1- \mathcal{T}\right)$. Thus they generate
the same hierarchy of symmetries starting from the equation itself \eqref{1.1}.

In section \ref{Sec2bis}, we explicitly present a Hamiltonian operator
\begin{equation}
		\mathcal{H} = w^2 \left(\mathcal{S} w \mathcal{S} -
\mathcal{S}^{-1}w \mathcal{S}^{-1} \right) w^2+ w^2\left(u_1-\mathcal{S}^{-1}
u\right) \left(\mathcal{S} w - w \mathcal{S}^{-1}\right)^{-1} (u
\mathcal{S}-u_1) w^2 .
	\end{equation}
Equation \eqref{1.1} is Hamiltonian with respect to it.
To prove that $\mathcal{H}$ is Hamiltonian in theorem \ref{theo4}, we apply the formalism
of rational and pre-Hamiltonian operators described in \cite{4, 42}.
Combining the expression of the Hamiltonian operator
$\mathcal{H}$ and the Nijenhuis recursion operator $\mathfrak{R}$ we show that
 $\mathcal{K} = \mathfrak{R}\mathcal{H}$ is also a Hamiltonian operator compatible with $\mathcal{H}$.
 Recursively we can prove the same result for all $\mathfrak{R}^k\mathcal{H}, k\in \mathbb{N}.$
 To increase the readability, we begin with recalling some basic definitions such as symmetries, recursion operators and Hamiltonian operators for integrable D$\Delta$Es, and
 end with a short discussion on the relevant open problems.

\section{Basic definitions for differential-difference equations}\label{Sec1}
	In this section we will give the definitions for some basic
concepts such as symmetries, recursion operators and Hamiltonian operators for integrable D$\Delta$Es.
By no means this is a complete overview of it. We only
mention the properties and theorems that we are going to use later.
More details on these concepts and related theories for both differential and differential-difference equations
can be found in \cite{52, 19} and a recent book \cite{b12}.

Let $u$ be  a function of a discrete variable $n \in
\mathbb{Z}$ and of a continuous variable $t$.
An evolutionary differential-difference
equation of dependent variable $u$ is of the form
	\begin{equation}
		u_t = f(u_p,\ldots,u_q),\qquad p\le q,\ \ p,q\in\mathbb{Z}.
		\label{1.3}
	\end{equation}
It is an abbreviated form to encode the infinite
sequence of ordinary differential systems of equations
\begin{equation*}
 \partial _t u(n,t)=f(u(n+p,t),\ldots,u(n+q,t)),\qquad
n\in\mathbb{Z}.
\end{equation*}
For convenience we drop the subscript when it is zero as $u_0 = u$.
We say function $f$ is of order $(p, q)$ if $\partial_{u_p} f\neq 0$ and $\partial_{u_q} f\neq 0$.
The difference between $p$ and $q$, $q-p$, is referred to as the total order of $f$. For example, equation \eqref{1.1} is of order $(-2, 2)$
with total order $4$.

All such functions, depending on a finite number of shifts of $u$, form a difference ring denoted $\mathcal{A}$ with the shift operator $\cS$ as its
automorphism defined by
$$
\cS:a(u_p, \cdots, u_q)\mapsto a(u_{p+1}, \cdots, u_{q+1}),\quad \cS:\alpha\mapsto \alpha; \quad a\in \cA,  \alpha\in \C.
$$
The reflection $\cT$ of the lattice $\Z$
defined by
\begin{equation}\label{autoT}
\cT:a(u_p, \cdots, u_q)\mapsto a(u_{-p}, \cdots, u_{-q}),\quad \cT:\alpha\mapsto \alpha; \quad a\in \cA,  \alpha\in \C
\end{equation}
is another automorphism of $\cA$. The composition $\cS\cT\cS\cT= {\rm Id}$ is the identity map.

On a difference ring $\mathcal{A}$ we can define derivations. A derivation (vector field) is said to be evolutionary if it commutes with the shift operator
$\cS$. Such vector field is completely determined by a function $g \in \mathcal{A}$. We call it the
characteristic of the vector field. Evolutionary equation \eqref{1.3} corresponds to an evolutionary vector field of characteristic $f$.

For any two evolutionary vector fields with characteristics $f$ and $g$, we define a Lie bracket
as follows
\begin{equation*}
		[f,g] = g_*(f)-f_*(g),
		\label{r5}
	\end{equation*}
where $
		f_* = \sum_{j \in \mathbb{Z}} \frac{\partial f}{\partial u_j}
\mathcal{S}^j $
is the Fr\'echet derivative of a function $f$.  We say that an evolutionary vector field with characteristic $g\in \mathcal{A}$ is a symmetry of equation \eqref{1.3}
if and only if $[f, g ] = 0$. If an equation possesses infinitely many commuting higher order symmetries, it is said to be integrable.

Often the symmetries of integrable equations can be generated by recursion operators \cite{55}.
Roughly speaking, a recursion operator is a linear operator $\mathfrak{R}$ mapping a symmetry to
a new symmetry. For an evolutionary equation \eqref{1.3}, it satisfies
	\begin{equation}
		\mathfrak{R}_*[f]=[f_*, \mathfrak{R}],
		\label{1.4}
	\end{equation}
	where $\mathfrak{R}_*[f]$  is the Fr\'echet derivative of $\mathfrak{R}$ along the evolutionary vector field  of characteristic $f$.

Recursion operators for nonlinear integrable equations are often Nijenhuis operators, that is,
\begin{equation}
		[\mathfrak{R} a, \mathfrak{R} b] - \mathfrak{R} [\mathfrak{R} a, b] - \mathfrak{R}[a, \mathfrak{R} b] + \mathfrak{R}^2[a, b] = 0, \quad \mbox{ for all $a, b \in \cA$.}
		\label{Nijen1}
	\end{equation}
The Nijenhuis property ensures that recursion operators generate commuting symmetries.

We now introduce briefly the concept of Hamiltonian operators for D$\Delta$E.
For any element $h \in \mathcal{A}$, we define an equivalent class (or a functional) $\int h$ by
saying that two elements $h, g \in \mathcal{A}$ are equivalent if $h -g \in {\rm Im}(\cS -1)$. The space of functionals
is denoted by $\mathcal{A}'$, which does not inherit a ring structure of $\mathcal{A}$.
For any functional $\int h \in \cA'$ (simply written $h \in \cA'$ without confusion), we define its difference variational derivative (Euler operator) denoted by
$\delta_u f$ as
	\begin{equation*}
		\delta_u f  =\sum_{i \in \mathbb{Z}} \mathcal{S}^{-i}
\frac{\partial }{\partial u_i} f .
	\end{equation*}
We say an evolutionary equation \eqref{1.3} is a
Hamiltonian equation if there exists a Hamiltonian operator $\cH$ and a Hamiltonian $h\in \cA'$ such that
	\begin{equation*}
		u_t = \mathcal{H}\, \delta_u  h .
		\label{Ham1}
	\end{equation*}
A difference operator $\cH$ is called a Hamiltonian if the bracket defined on $\cA'$ by
$\{\int\!\! f, \int\!\! g\}=\int\! \delta_u f\ \cH \delta_u g$ is a Lie bracket, that is, the bracket is skew-symmetric and satisfies the Jacobi identity.
Often Hamiltonian operators are rational instead of local.
We will follow the pre-Hamiltonian formalism developed in \cite{4, 42} to give the definition.

	A difference operator $A$ is said to be pre-Hamiltonian if there exists a two form on $\cA$ denoted by $\omega_A$ such that
	\begin{equation}
		A_*[A a](b) - A_*[A b](a) = A (\omega_A (a,b)), \quad \mbox{ for all $a, b \in \cA$.}
		\label{Ham4}
	\end{equation}
The Fr\'echet derivative of a difference operator $A(\cdot)$ is a bi-difference
operator $A_*[\cdot](\cdot)$, by convention we indicate with round parenthesis
the argument of the operator itself and with square bracket the direction of the Fr\'echet derivative determined by the characteristic of the evolutionary vector field.

Let $\cH=A B^{-1}$ be a skew-symmetric rational operator, where both $A$ and $B$ are difference operators. We say that $\cH$ is Hamiltonian if
$A$ is pre-Hamiltonian with $\omega_A$ determined by \eqref{Ham4} and $B$ satisfies
	\begin{equation}
		B_*[A a](b)-B_*[A b](a)+B_*(a)^\dagger[A b]+A_*(a)^\dagger[B b]
= B (\omega_A(a,b)) ,
		\label{Ham5}
	\end{equation}
where $A^\dagger$ is the adjoint of the difference operator $A$.
Given two pre-Hamiltonian operators $A,B$ they form a pre-Hamiltonian pair
$(A,B)$ if any linear combination $A + \zeta B$ of them is a pre-Hamiltonian
operator for any number $\zeta \in \mathbb{C}$.

Pre-Hamiltonian operators are also closely linked with Nijenhuis
operators.
	It has been proven in \cite{4, 41} that any Nijenhuis operator that can be
written in a minimal rational decomposition $\mathfrak{R} = A B^{-1}$ with $B$
pre-hamiltonian generates a pre-Hamiltonian pair $(A,B)$ and, vice-versa, any
pre-hamiltonian pair $(A,B)$ generates a Nijenhuis operator as $\mathfrak{R} = A B^{-1}$.
This provides an efficient way to prove that a certain rational operator is
Nijenhuis by checking if its factors form a pre-Hamiltonian pair, instead of vanishing of
the Nijenhuis torsion \eqref{Nijen1}.

\section{Construction of recursion operators}\label{Sec2}
For an integrable evolutionary equation, if it is known its Lax representation, there is a powerful approach to construct its recursion operator \cite{48}.
This idea has been extended for Lax
pairs that are invariant under the reduction groups \cite{7,1}.
In this section,  we use the same idea to construct a recursion operator
of an equation \eqref{1.1} from its Lax representation.

We rewrite the scalar Lax representation
\eqref{lax} for $p=2$ in the matrix form with the auxiliary system
	\begin{equation*}
		\mathcal{S} \Phi = U(\lambda) \Phi \qquad \Phi_t = V(\lambda) \Phi,
		\label{newaux}
	\end{equation*}
where $\lambda \in \mathbb{C}$ is the spectral parameter and
	\begin{subequations}
		\begin{gather}
			U(\lambda)=\begin{pmatrix}
				\frac{u_1}{u w_1} & \frac{\lambda}{w_1} &
-\lambda \frac{u_1 w}{u w_1}\\
				1 & 0 & 0\\
				0 & 1 & 0
			\end{pmatrix} \label{2.71}\\
			V(\lambda)=\begin{pmatrix}
				\lambda +\frac{u_2 w_1}{u} & -u u_1  &
-\left(w+\lambda\frac{u_2 w w_1}{u}\right)\\
				\frac{u_1 w^2}{u} +
\frac{1}{\lambda}\frac{u_{-1} w}{u} & \lambda - \frac{1}{\lambda} &
-\left(\frac{u_{-1} w^2}{u} + \lambda\frac{u_1 w}{u}\right)\\
				w+\frac{1}{\lambda}\frac{u_{-2} w w_{-1}}{u}& u
u_{-1} & -\left(\frac{1}{\lambda} +\frac{u_{-2} w_{-1}}{u}\right)
			\end{pmatrix}. \label{2.72}
		\end{gather}
		\label{2.7}
	\end{subequations}
The compatibility condition of $U(\lambda)$ and $V(\lambda)$, simply writing as $U$ and $V$, leads to
	\begin{equation*}
		U_t = \mathcal{S}(V)U - U V.
		\label{2.5}
	\end{equation*}
Notice that matrices $U(\lambda)$ and $V(\lambda)$ are invariant with respect to the
transformation:
	\begin{equation}
		 V(\lambda)=  - \mathcal{J} \mathcal{T} V(\lambda^{-1})
\mathcal{T} \mathcal{J},\qquad \cS^{-1} U^{-1}(\lambda)\cS=  \mathcal{J} \cT U(\lambda^{-1}) \cT  \mathcal{J},
		\label{2.9bis}
	\end{equation}
where the reflection operator $\mathcal{T}$ is defined by \eqref{autoT}.
	The symbol $\mathcal{J}$ stands for a ``reflected'' identity matrix
	\begin{equation*}
		\mathcal{J} = \begin{pmatrix}
			0 & 0 & 1\\
			0 & 1 & 0\\
			1 & 0 & 0\\
		\end{pmatrix},
		\label{2.8tris}
	\end{equation*}
which is sometimes called the backward identity or
exchange matrix. The transformation (\ref{2.9bis}) reflects the symmetry $\cT(u_t)=-u_t$ of the
equation (\ref{1.1}). The primitive automorphic function
of the group generated by the transformation $\lambda\mapsto \lambda^{-1}$ is
	\begin{equation}
		\mu(\lambda) = \lambda + \lambda^{-1}
		\label{prim}
	\end{equation}
For a given matrix $U$, we can build up a hierarchy of nonlinear systems by
choosing a matrix $V'$ with the degree of $\lambda$ from $-l$ to $l$.
The way to construct a recursion operator directly from a Lax representation
is to relate the different operators ${\bar V'}$ using the ansatz $${\bar
V'}=\mu (\lambda) V' +B$$ and then to find the relation between the two flows corresponding to
$\bar V'$ and $V'$. Here $B$ is the remainder and we assume that it
has the same symmetry as $V$, and thus is of the form:
	\begin{equation}
		B=\begin{pmatrix}
			b^{(1)}_{    11} & b^{(1)}_{    12} & b^{(1)}_{    13}\\
			b^{(1)}_{    21} & b^{(1)}_{    22} & b^{(1)}_{    23}\\
			b^{(1)}_{    31} & b^{(1)}_{    32} & b^{(1)}_{    33}
		\end{pmatrix} \lambda +
		\begin{pmatrix}
			b^{(0)}_{    11} & b^{(0)}_{    12} & b^{(0)}_{    13}\\
			b^{(0)}_{    21} & b^{(0)}_{    22} &
-\mathcal{T}b^{(0)}_{    21}\\
			-\mathcal{T}b^{(0)}_{    13} & -\mathcal{T}b^{(0)}_{
12} & -\mathcal{T}b^{(0)}_{    11}
		\end{pmatrix}-
		\begin{pmatrix}
			\mathcal{T} b^{(1)}_{    33} &  \mathcal{T} b^{(1)}_{
  32} &\mathcal{T}b^{(1)}_{    31}\\
			\mathcal{T} b^{(1)}_{    23} &  \mathcal{T} b^{(1)}_{
  22} &\mathcal{T}b^{(1)}_{    21}\\
			\mathcal{T} b^{(1)}_{    13} &  \mathcal{T} b^{(1)}_{
  12} &\mathcal{T}b^{(1)}_{    11}
		\end{pmatrix} \lambda^{-1}
		\label{3.6b}
	\end{equation}
	We need to stress that, due to the action of the reduction group, there
is one more relation we should consider
	\begin{equation}
		b^{(0)}_{    22} = -\mathcal{T}b^{(0)}_{    22}.
		\label{5.1bis}
	\end{equation}
Assume that $V'$ associates to the symmetry flow $u_{\tau}$ and $\bar V'$ to the flow $u_t$.
The compatibility condition leads to
	\begin{equation}
		U_t = \mu(\lambda) U_\tau + S(B) U - U B .
		\label{3.4b}
	\end{equation}
	When we apply these assumptions to \eqref{3.4b}, we get a over-determined system counting
34 equations of 14 parameters. We proceed considering some of these equations to find a
representation of $\mathfrak{R}$ and at the same time checking the others as the
compatibility conditions to ensure the existence of a solution, which leads to the following statement:
	\begin{theorem}\label{theo2}
		A recursion operator of the equation \eqref{1.1} is given by (\ref{RecursionInIntro}).
\iffalse
		\begin{subequations}
			\begin{align}
				& \mathfrak{R} =
\mathfrak{R}^{(0)}+\mathfrak{R}^{(1)}+\mathfrak{R}^{(2)}\\
				& \mathfrak{R}^{(0)} =w^2\left(\mathcal{S}w
\mathcal{S} +\mathcal{S}^{-1} w \mathcal{S}^{-1} \right) \frac{1}{w^2} \\
				& \mathfrak{R}^{(1)} =w^2\left( {u_1}-
\mathcal{S}^{-1} u\right) \left(\mathcal{S}w - w \mathcal{S}^{-1}\right)^{-1}
\left(  u \mathcal{S} + u_1\right)\frac{1}{w^2} \\
				& \mathfrak{R}^{(2)} =w^2\left(\left(
{u_1}-\mathcal{S}^{-1} u \right) \left(\mathcal{S}w - w
\mathcal{S}^{-1}\right)^{-1} \left(u \mathcal{S}-u_1\right) + \mathcal{S}
w\mathcal{S} -  \mathcal{S}^{-1}w \mathcal{S}^{-1}\right)
u\left(\mathcal{S}^2+\mathcal{S }+1\right) (\mathcal{S}^2-1)^{-1}\frac{u}{w^2}
			\end{align}
			\label{3.13}
		\end{subequations}
\fi
	\end{theorem}
	\begin{proof}
		We consider the ansatz in \eqref{3.6b} and apply it to
\eqref{3.4b}. The result is a matrix system involving the spectral parameter
$\lambda$, that allow us to decompose it even further considering that every
coefficient of $\lambda$ needs to be equal to zero independently. By inspection,
some of these relations are trivial. We notice that two parameters are
directly zero
		\begin{equation}
			b_{21}^{(1)}=0 \qquad b_{31}^{(1)}=0
			\label{5.2}
		\end{equation}
		and some are connected with one to one relations, such as
		\begin{equation}
			b^{(1)}_{11}= \mathcal{S} b^{(1)}_{22} \qquad
b^{(1)}_{21}=\mathcal{S}b^{(1)}_{32}=0 .
			\label{5.3}
		\end{equation}
We consider \eqref{3.4b} as the matrix equation $\Delta_{
ij}=0$. Notice that, because of the structure of $U$, all the
equations coming from $i=2, 3$ and $j=1, 2, 3$ do not involve
$t$-derivative and they can be used to reduce the number of parameters.

First we express all parameters in $B$ in terms of 4 parameters, namely,
$b^{(1)}_{11},  b^{(1)}_{12},  b^{(1)}_{33}$ and $b^{(0)}_{22} $, and we have
\begin{equation}
 b^{(1)}_{11}= (\mathcal{S}^{-2}-1)^{-1}\frac{u}{w^2}u_\tau. \label{r11b1}
\end{equation}
We then find
$b^{(1)}_{33}$ and $b^{(1)}_{12}$ in terms of $b^{(1)}_{11}$ as follows:
		\begin{eqnarray}
&& b^{(1)}_{    33} =-\left(\frac{1}{u^2}
(\mathcal{S}^2-1)+\mathcal{S}^2+\mathcal{S}\right) \mathcal{S}^{-2} b^{(1)}_{11}\label{r11bis}\\
&&\left(\mathcal{S} w- w \mathcal{S}^{-1} \right)
b^{(1)}_{    1 2}= u u_1 \left( \left(\mathcal{S}+1\right)
\frac{1}{u^2}\left(\mathcal{S}^2-1\right)+\left(\mathcal{S}^3-1\right)\right)
\mathcal{S}^{-2} b^{(1)}_{11}\label{b121}
		\end{eqnarray}
Finally, we analyse the equations involving $u_t$ and get
		\begin{equation}
			\begin{split}
				\frac{1}{w^2} u_t= & -u\left(\mathcal{S}
^2+\mathcal{S} +1\right) \mathcal{\mathcal{S} }^{-1} b^{(0)}_{    22}   -\Big(\left(\cS w \cS + \cS^{-1} w \cS^{-1}\right)u^{-1} \left(\cS^2 -1\right)+\\ & + \left(\cS w \cS - \cS^{-1} w \cS^{-1}\right) u \left(\cS^2 + \cS + 1 \right)\Big)
\mathcal{\mathcal{S} }^{-2} b^{(1)}_{11} -\left(
{u_1}- \mathcal{S}^{-1} u\right) b^{(1)}_{    12},
			\end{split}
			\label{r12}
		\end{equation}
where the only unknown parameter left is $b^{(0)}_{22}$ satisfies
		\begin{eqnarray}
				&& \left(\left(\mathcal{S}^2+\mathcal{S}+1\right)
u^2+\mathcal{S}^2-1\right) b^{(0)}_{22} =\frac{{u_1} w}{u}\cS^{-1} \left(\mathcal{S} w - w
\mathcal{S}^{-1} \right) b^{(1)}_{12}\nonumber\\&&
\qquad \qquad
-\frac{ {u_1} w}{u} \cS^{-1} \left(u_1
{u} \left(\cS \frac{ {w}^2}{ {u}^2}-\frac{ {w}^2}{
{u}^2}\mathcal{S}^{-2}\right)  +
\left(\frac{u_1}{u}-\frac{u}{u_1}\mathcal{S}^{-1}\right) \right) b^{(1)}_{11} . \label{b220}
		\end{eqnarray}
Substituting \eqref{b121} into \eqref{b220}, we obtain its right-handed side equals zero,
which leads to
\begin{equation*}
b^{(0)}_{22} = 0 .
		\end{equation*}
We now substitute it together with (\ref{r11b1}) and (\ref{b121}) into \eqref{r12}, which leads to
the recursion operator (\ref{RecursionInIntro}) as stated.
\end{proof}
The recursion operator \eqref{RecursionInIntro} can be written in a form involving the reflection operator $\cT$:
we start with an equation for $u_t$ containing the
$\mathcal{T}$-reflection operator, different from \eqref{r12} and
we then solve all the other equations in the same way as before. In this way, we obtain the following form
which is equivalent to \eqref{RecursionInIntro}:
	\begin{equation*}
		\begin{cases}
			&u_t = w^2 (1-\mathcal{T}) u_1
\left(b^{(1)}_{12}+\frac{u_2 w_1}{u_1}
\left(\frac{1}{u_2^2}\left(\mathcal{S}^2-1\right)+\mathcal{S}^2+\mathcal{S}
+1\right) b^{(1)}_{11}\right)\\
			& b^{(1)}_{1 2}= \left( \mathcal{S} w - w
\mathcal{S}^{-1} \right)^{-1} u u_1 \left( \left( \mathcal{S}+1
\right)\frac{1}{u^2}\left( \mathcal{S}^2-1\right) + \mathcal{S}^3-1\right) \cS^{-2}
b^{(1)}_{    11} \\
			& b^{(1)}_{
11}=(\mathcal{S}^{-2}-1)^{-1}\frac{u}{w^2}u_\tau
		\end{cases}
		\label{5.6bis}
	\end{equation*}
which leads to (\ref{RT}).
\iffalse
\begin{equation}
		\begin{split}
			\mathfrak{R} & = w^2 \left(1-\mathcal{T}\right) u_1 \left( \mathcal{S} w - w \mathcal{S}^{-1} \right)^{-1}\Bigg[ \left( \frac{w_1 w_2}{u_2 u_3}\mathcal{S}^3 - \frac{1}{u_1 u} \cS +\frac{u_1}{u}\right)\left( \cS^2 -1\right) + \\ &+\left(\frac{u_3 w_1 w_2}{u_2}\cS^3-\frac{u_1}{u}\cS - u u_1 \right)
\left(\cS^2 + \cS + 1 \right) \Bigg] (\mathcal{S}^2-1)^{-1}\frac{u}{w^2} .
		\end{split}
		\label{3.13bis}
	\end{equation}
\fi
It is obvious that the whole hierarchy is anti-symmetric (odd) under
$\mathcal{T}$ operator as $\mathfrak{R} = w^2 (1-\mathcal{T})
\hat{\mathfrak{R}}$, where $\hat{\mathfrak{R}}$ is a rational operator of $\cS$.
The presence of $\mathcal{T}$ prevents us from directly exploiting a series of
theorems proved for rational operators of $\cS$ in \cite{4}. On the other hand, this
alternative and more compact form for the recursion operator shows an important property of the
symmetries that they are all odd with respect to $\mathcal{T}$.

One of important properties of recursion operators is to generate an infinite hierarchy of commuting symmetries, which
is guaranteed by their Nijenhuis property \eqref{Nijen1}. To show that $\mathfrak{R}$ defined by \eqref{RecursionInIntro} is a
Nijenhuis operator, instead of directly checking the vanishing of the Nijenhuis torsion,
we use the pre-Hamiltonian formalism, in particular Theorems 3
in \cite{4}: if two difference operators $A$ and $B$ form a preHamiltonian pair, then their ratio $AB^{-1}$ is Nijenhuis.
	\begin{theorem} \label{Theorem1}
		The recursion operator $\mathfrak{R}$ in \eqref{RecursionInIntro} is a
Nijenhuis operator
	\end{theorem}
	\begin{proof}
		The recursion operator in \eqref{RecursionInIntro} can be rewritten as follows:
		\begin{equation}
			\mathfrak{R} = w^2 \left(A + B C^{-1} D\right)
E^{-1}\frac{1}{w^2},
			\label{factorised2}
		\end{equation}
		where each symbol above corresponds to a difference operator
		\begin{equation}
			\begin{split}
				& A = \left(\left(\mathcal{S}w \mathcal{S}
+\mathcal{S}^{-1} w \mathcal{S}^{-1}
\right)\frac{1}{u}\left(\mathcal{S}^2-1\right) +\left( \mathcal{S} w\mathcal{S}
-  \mathcal{S}^{-1}w \mathcal{S}^{-1}\right) u
\left(\mathcal{S}^{2}+\mathcal{S}+1\right)\right)\\
				& B = \left( {u_1}-\mathcal{S}^{-1} u\right) \\
				& C = \left(\mathcal{S}w - w
\mathcal{S}^{-1}\right)\\
				& D =  u u_1\left( \left(  \mathcal{S}+1\right)
\frac{1}{u^2} \left( \mathcal{S}^2-1\right)
+\mathcal{S}^3-1\right)\\
				& E = \frac{1}{u}\left(\mathcal{S}^2-1\right)
			\end{split}
			\label{definition2}
		\end{equation}
		From this setting, it is possible to rearrange the terms such
that the recursion operator appears as a rational operator $\mathfrak{R} = M
N^{-1}$. In doing so we need first to simplify $C^{-1}D$ and we convert it into a right fraction. We know from \cite{4} that the ring of difference operators satisfies the Ore condition, by consequence we have
\begin{equation*}
\begin{split}
C^{-1}D = &  \left({u u_1}\right)^{-1}\left(u_{-1} u_1 w \mathcal{S}^2+u u_2 w_1 \mathcal{S}\right)\\
& + C^{-1} \left((u_1 u_{2}-u u_3 w_1 w_2) \mathcal{S}-(u_{-1} u-u_{-2} u_{1} w_{-1} w) \right)\\
 = &  \left({u u_1}\right)^{-1}\left(u_{-1} u_1 w \mathcal{S}^2+u u_2 w_1 \mathcal{S}\right) + G I^{-1}(u_1 u_{-1})^{-1} \mathcal{S}
\end{split}
\end{equation*}
where operators $G$ and $I$ are given by
\begin{equation}\label{II}
\begin{split}
&G=  p_1\left(p_2\left(p_1 w_1-p_3 w_2 \right)\mathcal{S}- p \left(p_1 w - p_{-1}w_{-1}\right)\right) \qquad \quad p= u u_{-1}-u_1 u_{-2} w w_{-1} \\
&I=  p_2 w\left(p_3 w_2-p_1 w_1\right)\mathcal{S}+\left(p_{-1}p_2 w_{-1} w_1 - p p_1 w^2\right) +p_{-1} w\left(p_{-2} w_{-2}-p w_{-1}\right)\mathcal{S}^{-1}
\end{split}
\end{equation}
This leads to
$$
M= w^2 \left(A I + B \left({u u_1}\right)^{-1}\left(u_{-1} u_1 w \mathcal{S}^2+u u_2 w_1 \mathcal{S}\right) I + B G\right), \qquad
N= w^2 E \mathcal{S}^{-1} u_1 u_{-1} I.
$$
We now prove that $M$ and $N$ form a preHamiltonian pair by checking that the operator
$M+\lambda N$ is preHamiltonian for any constant $\lambda$ using computer algebra Mathematica,
similar to what we did for Theorem 7 in \cite{4}. Here we only present how we prove that
the operator $N$ is a preHamiltonian operator.
We notice that the first part of $N$ given by
		\begin{equation}
            \frac{w^2}{u}\left(
\mathcal{S}^2-1\right)\mathcal{S}^{-1} u_1 u_{-1}
			\label{Nijen2}
		\end{equation}
is a very simple pre-Hamiltonian operator, which we recognise via \eqref{Ham4} defining an associate bi-difference operator
		\begin{equation*}
			\omega(a,b) = \frac{u_{-3} u_1 w_{-1}^2 a
b_{-2}-u_{-3} u_1 w_{-1}^2 a_{-2} b +u_{-1} u_3 w_1^2 a_2 b -u_{-1} u_3 w_1^2 a
b_2}{u_{-1} u_1} .
			\label{Nijen3}
		\end{equation*}
		According to Lemma 3 in \cite{42}, given a generic
pre-Hamiltonian difference operator $P$ with associated bi-difference operator
$\omega_P$, then the operator $P Q$, where $Q$ is a difference operator,
is itself pre-Hamiltonian if and only if the following expression is in the image of
the operator $Q$, i.e.,
		\begin{equation*}
			Q_*[P(Q(a))](b)-Q_*[P(Q(b))](a)+\omega_P(Q(a),Q(b)) \in
\Ima\left(Q\right)
		\end{equation*}
Using computer algebra Mathematica,  we verify this property on operators
\eqref{Nijen2} and $I$ given by \eqref{II}. Thus prove that the operator $N$ is a preHamiltonian operator.
	\end{proof}
\section{Hamiltonian operator}\label{Sec2bis}
	In this section we will show that \eqref{1.1} is a Hamiltonian equation
	via the pre-Hamiltonian theory \cite{4,42}. Let
	\begin{equation}
		\mathcal{H} = w^2 \left(\mathcal{S} w \mathcal{S} - \mathcal{S}^{-1}w
\mathcal{S}^{-1} \right) w^2+ w^2\left(u_1-\mathcal{S}^{-1} u\right)
\left(\mathcal{S} w - w \mathcal{S}^{-1}\right)^{-1} (u \mathcal{S}-u_1) w^2 .
		\label{Hamiltonian2}
	\end{equation}
Note that \eqref{1.1} can be written as $u_t=\mathcal{H} \delta_u \ln(w) $.
We are going to prove that \eqref{Hamiltonian2} is a
Hamiltonian operator.
	\begin{theorem}\label{theo4}
		The operator $\mathcal{H}$ given by \eqref{Hamiltonian2} is
a Hamiltonian operator.
	\end{theorem}
	\begin{proof}
		We prove the statement by a direct computation. For simplicity, we write $\cH$ as follows:
		\begin{equation}
			\cH = w^2 \left( F - B C^{-1} B^\dagger \right)w^2,
			\label{HamilDeco}
		\end{equation}
where operators $B, C$ are defined by \eqref{definition2} and  $F$  is defined as
\begin{equation}
   F = \mathcal{S} w \mathcal{S} - \mathcal{S}^{-1}w
\mathcal{S}^{-1}.
\label{Definitions2}
\end{equation}
To show that $\cH$ is a Hamiltonian operator, we
will use the definition given in Section \ref{Sec1}: in particular, we are going to verify identities
\eqref{Ham4} and \eqref{Ham5}. We first note that $\cH$ is skew-symmetric since
$\cH + \cH^\dagger = 0$.
We then need to express $\cH$ in \eqref{Hamiltonian2} as a rational operator. Note that
\begin{subequations}
 \begin{align*}
  & C^{-1} B^\dagger = G I^{-1}\qquad\qquad\qquad\qquad\quad\qquad\;\; p = u_{-1} u w_1 - u_1 u_2 w\\
  & G = u_{-1} w w_1^2 w_2 p_1 \cS -u_2 w_{-1} w^2 w_1  p_{-1} \\
  & I =  w_1^2 w_2 p_1  \cS+ w_{-1} w_1 \left(u^2 w_{-1} w_1-u_{-2} u_2 w^2\right) - w_{-2} w_{-1}^2 p_{-2} \cS^{-1}
 \end{align*}
\end{subequations}
 As we did in theorem \ref{Theorem1}, due to the Ore property of the ring of the difference operators \cite{4}, we obtain
\begin{equation*}
  \cH =  M N^{-1}, \quad M = w^2\left(F I + B G\right); \quad N=w^{-2} I.
\end{equation*}
We now prove that the operator $M$ is pre-Hamiltonian using
\eqref{Ham4}. The calculations are straightforward though not so neat. The two form
 $\omega_M$ is skew-symmetric and of order $(-5,5)$ for any $a, b\in \cA$. Here we explicitly write down the terms in $\omega_M(a,b)$
 involving fifth order shifts for both $a$ and $b$:
\begin{equation*}
    \omega_M^{(5)}(a,b)=-\frac{u_1 w w_2^2 w_3 w_4^2 w_5^2
w_6 \left(u_6 u_7 w_5-u_4 u_5 w_6\right)\left(a_5 b-a b_5\right) }{u_1 u_2
w-u_{-1} u w_1}
\end{equation*}
The last step is straightforward to verify \eqref{Ham5} by substituting
 the expressions of $M$, $N$ and $\omega_M$ into it, which is verified by using
computer algebra Mathematica.
\end{proof}
By now we have obtained a recursion operator
$\mathfrak{R}$ \eqref{RecursionInIntro}, which is Nijenhuis and a Hamiltonian operator
$\mathcal{H}$ \eqref{Hamiltonian2} for equation \eqref{1.1}. We will go further looking for Hamiltonian operators compatible with $\mathcal{H}$. Similar to pre-Hamiltonian operators, two
Hamiltonian operators $\mathcal{H},\mathcal{K}$ are said to be compatible and
form a Hamiltonian pair $\left(\mathcal{H},\mathcal{K}\right)$ if any their linear
combination $\mathcal{H}+\lambda \mathcal{K}$ is a Hamiltonian operator
for any constant $\lambda$.
Theorem 4 in \cite{42} states that the operator
$\mathcal{K} = \mathfrak{R} \mathcal{H}$ obtained from Nijenhuis recursion operator
$\mathfrak{R}$ in \eqref{RecursionInIntro} and the Hamiltonian operator $\mathcal{H}$ in
\eqref{Hamiltonian2}, is also Hamiltonian and compatible
with $\mathcal{H}$ if $\mathcal{K}$ is skew-symmetric.
	\begin{corollary}
		The operator $\mathcal{K} = \mathfrak{R} \mathcal{H}$ obtained
from the recursion operator in \eqref{RecursionInIntro} and the Hamiltonian operator in
\eqref{Hamiltonian2} is a Hamiltonian operator and it is compatible with
$\mathcal{H}$.
	\end{corollary}
	\begin{proof} From \eqref{factorised2} and \eqref{HamilDeco} we have
		\begin{equation}
			\mathcal{K} = \mathfrak{R}\mathcal{H}=w^2 \left(A + B^\dagger C^{-1} D\right)
E^{-1}\left(F - B C^{-1} B^\dagger\right) w^2,
			\label{previousexpression}
		\end{equation}
where we use the same definitions for operators involved as in \eqref{definition2} and in \eqref{Definitions2}.
It follows from Theorem 4 in \cite{42} that we only need to
check that the operator $\mathcal{K}$ is skew-symmetric to show that
it is a Hamiltonian operator compatible
with $\mathcal{H}$, that is, to show that $\mathcal{K}+\mathcal{K}^\dagger=0$.
		The expression \eqref{previousexpression} contains $C^{-1}$, we split the skew-symmetry condition, according to the
positions of $C^{-1}$, into four parts denoted by  $J^{(i)}, i=1, \cdots, 4$ as follows:
		\begin{subequations}
			\begin{align}
				& \mathcal{K} + \mathcal{K}^\dagger =
J^{(1)}+J^{(2)}+J^{(3)}+J^{(4)} \label{Jys1}\\
				& J^{(1)} = 2 w^2 B C^{-1}
\left(w^2-w_1^2\right) C^{-1} B^{\dagger} w^2\\
				& J^{(2)} = 2 w^2 B C^{-1} \left( u_1 w_1
\mathcal{S}^2- u w \mathcal{S}^{-1} \right) w^2 \label{Jys2}\\
				& J^{(3)} = 2 w^2 \left(u_1 w_1
\mathcal{S}-u_{-1}w_{-1} \mathcal{S}^{-2}\right) C^{-1} B^\dagger \label{Jys3}\\
				& J^{(4)} = 2 w^4 \left( w_{-1}^2 -w_1^2\right)
			\end{align}
			\label{Jys}
		\end{subequations}
Let us define the operators $\Bar{B}$ and $\Bar{C}$ similarly to $B$ and
$C$ in \eqref{previousexpression}
		\begin{equation*}
			\Bar{B} = u_1 + \mathcal{S}^{-1} u;\qquad \Bar{C} =
\mathcal{S} w + w \mathcal{S}^{-1}.
		\end{equation*}
Then we obtain the following identities
		\begin{subequations}
			\begin{align*}
				&  u_1 w_1 \mathcal{S}^2- u w \mathcal{S}^{-1} =
\frac{1}{2} \left( C \Bar{B}^\dagger-\Bar{C} B^\dagger \right);\\
				&  u_1 w_1 \mathcal{S} - u_{-1}w_{-1}
\mathcal{S}^{-2} = \frac{1}{2} \left(B \Bar{C}+\Bar{B} C\right),
			\end{align*}
		\end{subequations}
Substituting them into  $J^{(2)}$ \eqref{Jys2} and $J^{(3)}$ \eqref{Jys3}, respectively,  \eqref{Jys1} becomes
		\begin{equation*}
			\mathcal{K} + \mathcal{K}^\dagger  = w^2 B C^{-1}
\left(w^2-w_1^2 + \frac{1}{2} \left( C \Bar{C} - \Bar{C} C \right)\right) C^{-1}
B^{\dagger} w^2
		\end{equation*}
		which, by direct computation is zero, proving that $\mathcal{K}$
is skew-symmetric. Thus  operator $\mathcal{K}$ is a Hamiltonian operator
compatible with $\mathcal{H}$ defining a Hamiltonian pair
$(\mathcal{H},\mathcal{K})$.
\end{proof}
An immediate consequence of this result is that the operator
$\mathcal{R}^n\mathcal{H}$ is a Hamiltonian operator for all
$n \in \mathbb{Z}$. Thus every symmetry in the hierarchy for equation \eqref{1.1} is also Hamiltonian. Indeed, we can write its next symmetry \eqref{5.9} as
$$u_{t_2}=\mathcal{H} \delta_u u_1 u_{-1} w = \mathcal{R}\mathcal{H} \delta_u \ln(w) .$$

\section{Discussion and Further Work}\label{Sec4}
The main focus of this paper is to study the algebraic properties of the integrable equation
\eqref{1.1}, that is,
	\begin{equation*}
		u_t = w^2\left(u_2 w_1 - u_{-2}w_{-1}\right), \quad w = \sqrt{1+u^2} .
		\label{finale}
	\end{equation*}
We construct its recursion operator \eqref{RecursionInIntro} based on its Lax representation, and show that it is a Hamiltonian system with Hamiltonian operator $\cH$ defined by \eqref{Hamiltonian2}. Moreover, we obtain a hierarchy of Hamiltonian operators compatible  to $\cH$ following from the Nijenhuis property of the recursion operator.

There are the verifiable conditions to guarantee that a weakly nonlocal operator generates a local hierarchy starting from a proper seed \cite{82}. The recursion operator \eqref{RecursionInIntro} is not weakly nonlocal. We have not proven that it produces local symmetries even though the straightforward computation shows that starting from the equation itself, that is true for finite steps. A similar study as in \cite{4} is required in this case.

As mentioned in Introduction, the operator \eqref{RT} involving the reflection operator $\cT$ can be used to produce the hierarchy of infinite symmetries starting from equation itself just as  the recursion operator \eqref{RecursionInIntro}. This is true for other scalar integrable differential-difference equations. For instance, for the Volterra chain
$$
u_t=u (u_1-u_{-1}),
$$
its recursion operator is
$$
\mathcal{R}= u \cS +u_1 +u +u \cS^{-1} +u (u_1-u_{-1}) (\cS-1)^{-1} \frac{1}{u} .
$$
The corresponding $\mathcal{R}'$ is
$$
\mathcal{R}'= u (1-\cT) (\cS +1) u \cS (\cS-1)^{-1} \frac{1}{u}.
$$
We'll explore more to uncover whether there are any advantages to present in this form.

The equation \eqref{initialeq} can be written in a rational form
\begin{equation*}
 v_t=(1+v^2) \left( \prod_{k=1}^{p-1}\frac{1+v_k^2}{1-v_k^2} \frac{v_p}{1-v_p^2}-\prod_{k=1}^{p-1}\frac{1+v_{-k}^2}{1-v_{-k}^2}
 \frac{v_{-p}}{1-v_{-p}^2}\right),  \quad  p\in \mathbb{N}
\end{equation*}
via a point transformation \cite{96}\footnote{We would like to thank a reviewer to point out the reference.}
\begin{equation*}
 u = \frac{2 v}{1-v^2}.
\end{equation*}
All results obtained in terms of the variable $u$ including the Lax representation, recursion operators and Hamiltonian structures can be naturally passed to the variable $v$.

In principle, the methods developed in this paper can be used to construct a recursion operator for the entire family of integrable equations \eqref{squareq}. For any fixed $p$, we can write down the recursion operator although it is formed by more parts comparing to \eqref{RecursionInIntro}. We  have also done it for $p=3, 4$.  Until now we are not able to present it in a neat way as it has been done for the Narita-Itoh-Bogoyavlensky lattice in \cite{7}.

\subsection*{Acknowledgements}
This article is partially based upon work from COST Action CaLISTA CA21109 supported by COST (European Cooperation in Science and Technology). www.cost.eu.

\label{lastpage}
\end{document}